\documentclass[10pt,twocolumn]{IEEEtran}

\usepackage{tabularx}

\usepackage{algpseudocode}
\usepackage{algorithm}
\usepackage{algorithmicx}
\usepackage{lipsum} 
\usepackage{arydshln}
\usepackage[dvips]{color}
\usepackage{comment}
\usepackage{todonotes}
\usepackage{epsf}
\usepackage{epsfig}
\usepackage{times}
\usepackage{epsfig}
\usepackage{graphicx}
\usepackage{bbold}
\usepackage{mathtools}
\usepackage{mathrsfs}
\usepackage{amssymb}
\usepackage{pdfpages}
\usepackage{epstopdf}
\newfloat{algorithm}{t}{lop}

\usepackage{amsmath}

\usepackage{dsfont}
\usepackage{lettrine}
\usepackage{amsmath,epsfig,amssymb,algorithm,algpseudocode,amsthm,cite,url}

\usepackage{subcaption}
\allowdisplaybreaks
\usepackage{csquotes}

\usepackage{verbatim}

\usepackage[english]{babel}

\usepackage{mathtools, cuted}
\usepackage{amsmath,amssymb}
\DeclareMathOperator{\E}{\mathbb{E}}

\usepackage[justification=centering]{caption}
\usepackage{verbatim}

\newtheorem{theorem}{\bf Theorem}

\newtheorem{proposition}{\bf Proposition}

\begin{document}
	\pagenumbering{gobble}
	%
	\title{Generative Adversarial Networks for Distributed Intrusion Detection in the Internet of Things}
	\IEEEoverridecommandlockouts
	\author{\IEEEauthorblockN{Aidin Ferdowsi and Walid Saad\\}
		\IEEEauthorblockA{			Wireless@VT, Bradley Department of Electrical and Computer Engineering, \\ Virginia Tech, Blacksburg, VA, USA,
			Emails: \{aidin,walids\}@vt.edu\\}\vspace{-10mm}
		\thanks{This research was supported by the U.S. National Science Foundation under Grants ACI-1541105 and IIS-1633363.}
	}
	\maketitle
	

	%
	\IEEEpeerreviewmaketitle
	
	\begin{abstract}
		To reap the benefits of the Internet of Things (IoT), it is imperative to secure the system against cyber attacks in order to enable mission critical and real-time applications. To this end, intrusion detection systems (IDSs) have been widely used to detect anomalies caused by a cyber attacker in IoT systems. However, due to the large-scale nature of the IoT, an IDS must operate in a distributed manner with minimum dependence on a central controller. Moreover, in many scenarios such as health and financial applications, the datasets are private and IoTDs may not intend to share such data. To this end, in this paper, a distributed generative adversarial network (GAN) is proposed to provide a fully distributed IDS for the IoT so as to detect anomalous behavior without reliance on any centralized controller. In this architecture, every IoTD can monitor its own data as well as neighbor IoTDs to detect internal and external attacks. In addition, the proposed distributed IDS does not require sharing the datasets between the IoTDs, thus, it can be implemented in IoTs that preserve the privacy of user data such as health monitoring systems or financial applications. It is shown analytically that the proposed distributed GAN has higher accuracy of detecting intrusion compared to a standalone IDS that has access to only a single IoTD dataset. Simulation results show that, the proposed distributed GAN-based IDS has up to 20\% higher accuracy, 25\% higher precision, and 60\% lower false positive rate compared to a standalone GAN-based IDS. 
	\end{abstract}	 \vspace{-4mm}
	\section{Introduction}
	The Internet of Things (IoT) will interconnect a plethora of devices that collect data from their environment, process such data, and transmit valuable information to end-user applications such as health monitoring systems, drone delivery systems, self driving cars, and smart grids\cite{8660516,Ferdowsi,7123563,saad2019vision,ZARPELAO201725}. However, an effective deployment of these diverse IoT services requires an efficient intrusion detection system (IDS) that assures secure and reliable data transmission from IoT devices (IoTDs) by detecting anomalous activities such as sending data which is different from their normal state \cite{ZARPELAO201725}. 
	
	The use of IDS solutions has been studied extensively in IoT networks\cite{LIAO201316,Liu,Mitchell2014,Lee2014,Summerville,Misra,info7020025,Amaral}. For instance, the works in \cite{LIAO201316} and \cite{Liu} proposed signature-based methods in which the network behavior is compared with a set of attack signatures which is stored in a database. Moreover, the authors in \cite{Mitchell2014,Lee2014,Summerville} have proposed anomaly-based IDSs in which the activities of a network are compared to the normal behavior of the system and an alarm is triggered whenever the deviation from the normal state exceeds a threshold. Furthermore, new techniques are proposed in \cite{Misra,info7020025,Amaral} to compare the state of the system with predefined specifications such as the maximum capacity of links, packet size, or abstract rules based on IoT traffic.
	
	However, due to several unique characteristics of IoT systems such as their scale and privacy-preserving nature, traditional IDSs such as those proposed in \cite{LIAO201316,Liu,Mitchell2014,Lee2014,Summerville,Misra,info7020025,Amaral} may not be effective. In traditional methods, the system administrator implements the IDS on computationally capable data centers or nodes that process the data from the entire IoT system or from a subset of neighboring nodes. However, due to the large-scale nature of the IoT, relying on centralized nodes to deploy IDS solutions exposes such nodes to attacks and, thus, eventually the rest of the IoT can also be vulnerable if the centralized IDS is compromised\cite{8509635}. In addition, most IDSs require the availability of large amounts of data points related to the normal state of the network and failure scenarios\cite{LIAO201316,Liu,Mitchell2014,Lee2014,Summerville,Misra,info7020025,Amaral}. However, in IoT systems, each IoTD may have a dataset that only contains a limited portion of the network's state. Moreover, in many applications such as health monitoring systems or financial activities, the end user may not intend to share its available dataset with the administrator of the system which makes the data-centered IDS solutions ineffective \cite{ZARPELAO201725}.

	Recently, a number of distributed IDS and attack detection methods have been proposed for improving IoT security in \cite{Thanigaivelan2016,Cervantes,Butun,Abeshu,DIRO2018761,RAJASEGARAR20141833}. The authors in \cite{Thanigaivelan2016} and \cite{Cervantes} proposed anomaly detection approaches in which every node monitors its neighbors and, in case of abnormal behavior of a neighbor, breaks the data link from that neighbor. The authors in \cite{Butun} presented the challenges related to distributed private datasets and cloud-centric IoT systems. The works in \cite{Abeshu} and \cite{DIRO2018761} developed distributed deep learning algorithms for cyber attack detection in fog-to-things systems. In \cite{RAJASEGARAR20141833}, a clustering method has been proposed that combines the measurements of IoTDs before sending a compact description of their data to other nodes so as to minimize the communication overhead. However, the proposed IDSs in \cite{LIAO201316,Liu,Mitchell2014,Lee2014,Summerville,Misra,info7020025,Amaral} are centralized and not applicable for large scale IoT systems. In addition, the distributed IDS methods in \cite{Thanigaivelan2016,Cervantes,Butun,Abeshu,DIRO2018761,RAJASEGARAR20141833} cannot be applied to IoTDs that seek to preserve the privacy of their data. 
	
	The main contribution of this paper is, thus, to propose a distributed generative adversarial network (GAN) architecture as an IDS for IoT systems. Recently, GANs emerged as an effective unsupervised anomaly detection approach in computer vision and time series applications \cite{Ravanbakhsh,Ravanbakhsh2,Schlegl,zenati2018efficient,intrator2018mdgan,li2018anomaly}. However, these computer vision and time-series works consider a centralized GAN that has access to all of the data points. In contrast, in an IoT, a centralized GAN may cause communication overhead and increase the vulnerability of the IoT to the attacks on central units. The proposed distributed GAN-based IDS enables the IoTDs to monitor their neighbors in order to detect intrusion with minimum dependence on a central unit. In particular, the proposed distributed IDS can detect any type of attack that manipulates the integrity of the IoTDs such as bad data injection attacks. Moreover, we analytically show the superiority of the proposed distributed GAN-based IDS in terms of accuracy and intrusion detection probability compared to a single standalone GAN IDS. Simulation results show that, for a daily activity dataset \cite{dataset}, the proposed distributed GAN-based IDS has up to 20\% higher accuracy, 25\% higher precision, and 60\% lower false positive rate compared to a standalone GAN-based IDS. To the best of our knowledge, this paper is the first to propose a distributed GAN architecture as an IDS in privacy preserving IoT systems.
	
	The rest of the paper is organized as follows. The IoT system model and GAN architecture are presented in Section \ref{sec:systemmodel}. The proposed distributed GAN-based IDS is described in Section \ref{sec:GAN}. The simulation results are studied in Section \ref{sec:simulations} and conclusions are drawn in \ref{sec:conc}.
	\section{System Model}\label{sec:systemmodel}
	Consider an IoT system composed of a set $ \mathcal{N} $ of $ n $ IoTDs in which each IoTD $ i $ owns a set of its previously transmitted data points,$ \mathcal{D}_i $, that follows a distribution $ p_{\text{data}_i}(x) $, where $x$ can be time series, financial records, or health monitoring datasets depending on the IoT application. We consider $ \mathcal{D}_i $ contains data points from the normal state of the IoTD in which there is no intrusion to the IoT. We also let $ \mathcal{D}_1\cup \mathcal{D}_2 \cup \dots \cup \mathcal{D}_n = \mathcal{D} $, where $ \mathcal{D} $ is the total available data with a distribution  $ p_\text{data} $. In this model, every IoTD $ i  $ tries to learn a \emph{generator} distribution $ p_{g_i} $ over its available dataset $ \mathcal{D}_i $ such that $ p_{g_i} = p_{\text{data}_i}  $ and use that distribution to detect an intrusion to the system. The intrusion in our system is any activity by an attacker that causes an IoTD to communicate data points which do not follow its data distribution $p_{d_i}$. In fact, if an IoTD knows the distribution of its own normal state, it can easily discriminate a data point that is not similar to the normal state distribution. To learn the distribution $ p_{g_i} $ at every IoTD $ i $, we define a prior input noise $ z $ with distribution $ p_{z_i}(z) $ and a mapping $ G_i(z,\boldsymbol{\theta}_{g_i}) $ from this random variable $ z $ to data space, where $ G_i $ is an artificial neural network (ANN) with parameters $ \boldsymbol{\theta}_{g_i} $. An ANN, usually, consists of artificial neurons and activation functions which map an input to an output. We define another ANN called \emph{discriminator} $ D_i(\boldsymbol{x},\boldsymbol{\theta}_{d_i}) $ for every IoTD $ i $ that receives a data point $ \boldsymbol{x} $ and outputs a value between $ 0 $ and $ 1 $. When the output of the discriminator is closer to $ 1 $, then the received data point is a normal state and when the output is closer to $ 0 $ then the received data is an anomaly at IoTD $ i $. 
	
	\begin{figure*}
		\centering
		\includegraphics[width=0.9\textwidth]{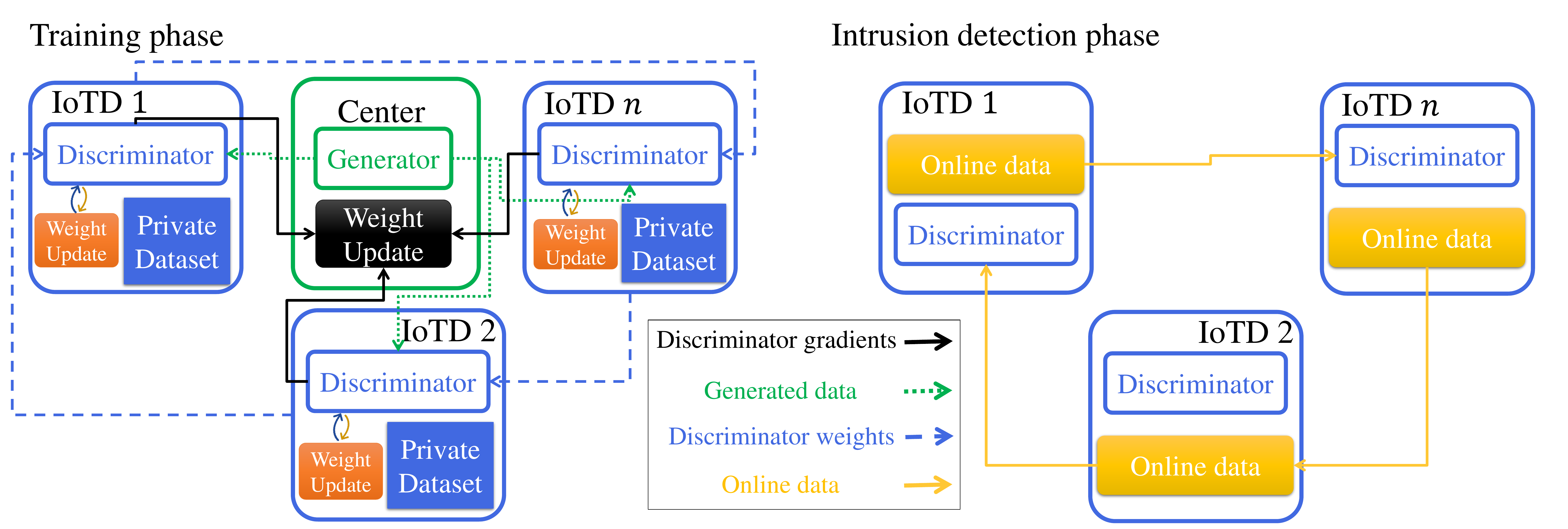}
		\caption{An illustration of the system model in training and intrusion detection phases.}
		\label{fig:systemmodel}
		\vspace{-5mm}
	\end{figure*}
	
	Every IoTD $i$'s generator ANN tries to generate data points close to the normal state data in order to find the best approximation of $p_{\textrm{data}_i}$. On the other hand, every IoTD's discriminator aims at discriminating the generated data points from its own dataset. In fact, the generated data points mimic anomalous state of the system because they are generated from a distribution $p_{g_i}$ that is not equal to $p_{\textrm{data}_i}$, because $\boldsymbol{\theta}_{g_i}$ are randomly chosen for an untrained ANN. Thus, the discriminator outputs 0 values for the generator \cite{NIPS2014_5423}. Therefore, the generator and discriminator at every IoTD $i$ interact to find the the optimal $\boldsymbol{\theta}_{g_i}$ and $\boldsymbol{\theta}_{d_i}$ such that the generator can generate data points close to the normal state while the discriminator can discriminate between the abnormal (anomalous) and normal data points. This interaction can be modeled using a game-theoretic framework \cite{NIPS2014_5423} in which we define a local \emph{value} function at every IoTD $ i $ as follows:
	\begin{align}\label{eq:localvalue}
		V_i (D_i,G_i) &= \mathbb{E}_{\boldsymbol{x}\sim p_{\textrm{data}}}\left[\log D_i(x)\right]\nonumber
		\\& +  \mathbb{E}_{z\sim p_{z_i}}\log\left( 1- D_i(G_i(z))\right).
	\end{align}
	The value function \eqref{eq:localvalue}, jointly quantifies how close is the generator's generated data points to the normal state and how good the discriminator can discriminate between the normal and abnormal data points.
	In \eqref{eq:localvalue}, the first term is defined to force the discriminator to produce values equal to 1 for real data. Meanwhile, the second term penalizes any anomalous point generated by the generators.
	While every IoTD's generator will seek to minimize the value function defined in \eqref{eq:localvalue}, the discriminator tries to maximize this value. Therefore, the optimal solutions for the discriminator and generator can be derived from the following minimax problem\cite{NIPS2014_5423}:
	\begin{align}\label{eq:minimax}
		\{D_i^*,G_i^* \}= \arg\min_{G_i}\arg\max_{D_i} V_i(D_i,G_i).
	\end{align}
	However, since every IoTD has access to only its own dataset, the optimization problem for a standalone IDS in which every IoTD has access only to its own dataset:
	\begin{align}\label{eq:minimaxstandalone}
	\{\bar{D}_i,\bar{G}_i\} = \arg\min_{G_i}\arg\max_{D_i} \bar{V}_i(D_i,G_i),
	\end{align}
	where 
	\begin{align}\label{eq:valuelocal}
		\bar{V}_i (D_i,G_i) &= \mathbb{E}_{\boldsymbol{x}\sim p_{\textrm{data}_i}}\left[\log D_i(x)\right]\nonumber
		\\& +  \mathbb{E}_{z\sim p_{z_i}}\log\left( 1- D_i(G_i(z))\right).
	\end{align}
	Next, we show that a standalone IDS cannot optimally detect intrusion. To this end, first, we derive the optimal value of a standalone IDS.
\begin{theorem}\label{Theorem:GAN}
	The optimal value for a standalone IDS will be:
	\begin{align}\label{eq:valuestandalone}
		V_i (\bar{D}_i,\bar{G}_i) = -\log(4) + s(p_{\textrm{data}_i}||p_{\textrm{data}}),
	\end{align}	
	where $s$ is the Jensen-Shanon divergence between distributions $ p_{\textrm{data}_i}$ and $p_{\textrm{data}} $.
\end{theorem}
\begin{proof}
	To find the optimal strategies of the generator and discriminator of a standalone IoTD we have to solve \eqref{eq:minimaxstandalone}.	 From \cite[Theorem 1]{NIPS2014_5423} we know that the solution of \eqref{eq:minimaxstandalone} is at $ \bar{p}_{g_i} = p_{\textrm{data}_i}$. Now from \eqref{eq:localvalue} we have:
	 \begin{align}
	 {V}_i (\bar{D}_i,\bar{G}_i) &= \int_{x}p_{\textrm{data}}(x)\log \bar{D}_i(x)dx\nonumber\\
	 &+ \int_{z}p_{z_i}(z)\log\left(1-\bar{D}_i\left(\bar{G}_{i}\left(z\right)\right)\right)dz\nonumber\\
	 & = \int_{x}\Big(p_{\textrm{data}}(x)\log \bar{D}_i(x)\nonumber\\
	 &+ \bar{p}_{{g}_i}(x)\log\left(1-\bar{D}_i\left(x\right)\right)\Big)dx. \label{eq:V_ibar}
	 \end{align}
	 For any $ p_{\bar{g}_i} $, \eqref{eq:V_ibar} reaches its minimum at $ \bar{D}_i = \frac{p_{\textrm{data}}}{p_{\textrm{data}} + \bar{p}_{g_i}} $. Since we know the standalone IoTD reaches its optimal point at $ \bar{p}_{g_i} = p_{\textrm{data}_i}$, then we will have:
	 \begin{align}
	 	{V}_i (\bar{D}_i,\bar{G}_i) &= \int_{x}\Bigg(p_{\textrm{data}}(x)\log \left(\frac{p_{\textrm{data}}(x)}{p_{\textrm{data}}(x) + p_{\textrm{data}_i}(x)}\right)\nonumber\\
	 	&+ p_{\textrm{data}_i}(x)\log\left(\frac{p_{\textrm{data}_i}(x)}{p_{\textrm{data}}(x) + p_{\textrm{data}_i}(x)}\right)\Bigg)dx\nonumber\\
	 	&= \int_{x}\Bigg(p_{\textrm{data}}(x)\log \left(\frac{p_{\textrm{data}}(x)}{2\frac{p_{\textrm{data}} (x) + p_{\textrm{data}_i}(x)}{2}}\right)\nonumber\\
	 	&+ p_{\textrm{data}_i}(x)\log\left(\frac{p_{\textrm{data}_i}(x)}{2\frac{p_{\textrm{data}} (x) + p_{\textrm{data}_i}(x)}{2}}\right)\Bigg)dx\nonumber\\
	 	&= \int_{x}\Bigg(p_{\textrm{data}}(x)\Bigg(\log \left(\frac{p_{\textrm{data}}(x)}{\frac{p_{\textrm{data}} (x) + p_{\textrm{data}_i}(x)}{2}}\right) \nonumber\\
	 	&+ \log(1/2)\Bigg)+ p_{\textrm{data}_i}(x)\Bigg(\log\left(\frac{p_{\textrm{data}_i}(x)}{\frac{p_{\textrm{data}} (x) + p_{\textrm{data}_i}(x)}{2}}\right)\nonumber\\
	 	&+ \log(1/2)\Bigg)\Bigg)dx,
	 \end{align}
	 which simplifies to \eqref{eq:valuestandalone}.
\end{proof}
From \cite[Theorem 1]{NIPS2014_5423} we know that the optimal value for a \emph{centralized} IDS which has access to all of the data points is $-\log(4)$. Thus, from Theorem \ref{Theorem:GAN}, we can derive that the standalone IDS cannot minimize the value function as good as a centralized IDS that is connected to all the data points since the optimal value of a standalone IDS is larger than that of a centralized IDS by a value of $s(p_{\textrm{data}_i}||p_{\textrm{data}})$. Next, we also show that the \emph{true positive rate}, $T_P$, i.e. correctly detected anomaly rate, of an standalone IDS is always less than a centralized IDS that has access to all of the data points.
\begin{proposition}\label{proposition}
	The $ T_P $ of a standalone IoTD anomaly detector is always less than an IoTD that has access to all of the data points.
\end{proposition}
\begin{proof}
	From \cite[Proposition 1]{NIPS2014_5423}, we know that the optimal discriminator of a centralized IDS that has access to the full dataset is $ D^*(x) = \frac{1}{2}, \, \forall x \in p_{\textrm{data}}$. Therefore, the difference between the discriminator of the standalone IoTD,$ \frac{p_{\textrm{data}}}{p_{\textrm{data}} + \bar{p}_{g_i}} $, and the discriminator of a GAN with the full data, $ 1/2 $, shows the error of the standalone IoTD. Thus, the $ T_P $ of the standalone IoT is complement of such error as follows:
	\begin{align}
		T_{P_i} &= \E \left[1 - \left|\frac{1/2 - \frac{ p_{\textrm{data}}(x)}{p_{\textrm{data}_i}(x) + p_{\textrm{data}}(x)}}{1/2}\right|\right],\label{eq:TP}
	\end{align}
	where \eqref{eq:TP} reaches its maximum when $p_{\textrm{data}_i}(x) = p_{\textrm{data}}(x)$, which contradicts with the fact that $p_{\textrm{data}_i}(x) \neq p_{\textrm{data}}(x)$. Therefore true positive ratio of an standalone IoTD is smaller than the IoT that has access to all data points.
\end{proof}
Although the standalone GAN may be able to detect internal intrusions, i.e., the IoTD can detect if an adversary has manipulated its own data, however, an attacker can also manipulate the IDS itself to stay stealthy at the attacked IoTD. One approach to stop the attacker from staying stealthy in internal intrusions can be to implement a centralized IDS that monitors all of the IoTDs. However, in this approach the communication overhead can be extremely high due to the large-scale nature of the IoT system. Moreover, in this case, the central IDS should have access to all of the IoTDs' datasets which may not be practical in a privacy preserving IoT system. Furthermore, a centralized IDS makes the IoT system vulnerable to attacks on the central unit. Therefore, next, we propose a novel distributed GAN architecture that provides an effective IDS by adapting a mechanism in which every IoTD monitors the neighbor IoTDs. 
\section{Distributed GAN-based IDS for IoT Systems}\label{sec:GAN}
In our proposed distributed GAN-based IDS, we build on the architecture from \cite{hardy2018md}. The goal of our distributed GAN is to find a discriminator at every IoTD without sharing their datasets with each other such that every IoTD's discriminator can discriminate if a new data point follows the total data distribution, $ p_\text{data} $. The main difference of the distributed IDS with the standalone IDS is that the standalone IDS learns to compare a new data point with its own data distribution, $ p_{\text{data}_i}  $, however, in the proposed distributed IDS, every IoTD can compare a new data point with the distribution of the total data $ p_\text{data} $. Therefore, since in the distributed IDS every IoTD's discriminator knows the distribution of the total data, then every IoTD can detect intrusion on other IoTDs as well. 

In our distributed GAN, \emph{during only the training phase}, we use a central unit that has a generator $G_{\boldsymbol{\phi}}$ where $\boldsymbol{\phi}$ is the weights of the generator's ANN. Furthermore, every IoTD only has a discriminator denoted by $D_{\boldsymbol{\theta}_i}$ where $\boldsymbol{\theta}_i$ is the weights of every discriminator's ANN. In this architecture, every IoTD is connected to at least one other IoTD in a wireless network such that the connection graph of the IoTDs must construct a \emph{cycle}. Moreover, at the training phase, every IoTD is connected to the central unit. We define $T$ as the period of \emph{epochs} at which IoTDs communicate with the center and $E$ as the period of epochs at which the IoTDs communicate with each other. An epoch is a training session during which all of the data points have been used to update the ANN weights. Next, we explain the training phase for the distributed GAN-based IDS architecture as shown in Fig. \ref{fig:systemmodel}.
\subsection{Training of the Discriminators}
Every $T$ epochs, for a fixed $\boldsymbol{\phi}$, the generator generates $2n$ batches $\left\{\mathcal{B}_{a_1},\dots,\mathcal{B}_{a_n}, \mathcal{B}_{g_1},\dots,\mathcal{B}_{g_n}\right\}$ of $b$ anomalous points, i.e. fake points that are not derived from the actual datasets. Moreover, the discriminator at each IoTD $i$ samples a batch $\mathcal{B}_{r_i}$ of $b$ points from its available dataset $\mathcal{D}_i$. The generator sends every generated batch $\mathcal{B}_{a_i}$ to each IoTD $ i $ which, in turn, calculates the following loss value:
\begin{align}
	L_i(\boldsymbol{\theta}_i) = \frac{1}{b}\left[\sum_{x \in \mathcal{B}_{r_i}}\log D_{\boldsymbol{\theta}_i}(x) + \sum_{x \in \mathcal{B}_{a_i}}\log\left( 1 - D_{\boldsymbol{\theta}_i}(x)\right) \right].
\end{align}
$L_i(\boldsymbol{\theta}_i)$ characterizes an approximation of the value function for each IoTD's discriminator in \eqref{eq:localvalue}.
Then using a gradient descent method such as the Adam optimizer \cite{kingma2014adam}, the IoTD updates its own weights, $ \boldsymbol{\theta}_i $.
\subsection{Training of the Central Generator}
Every $ T $ epochs, every IoTD $i$ uses $ \mathcal{B}_{g_i} $ to calculate the following loss:
\begin{align}
L_i^g(\boldsymbol{\theta}_i) = \frac{1}{b}\left[ \sum_{x \in \mathcal{B}_{g_i}}\log \left(1 - D_{\boldsymbol{\theta}_i}(x)\right) \right],
\end{align}
which is an approximation of the value function for each IoTD's generator in \eqref{eq:localvalue}. Note that, in \cite{NIPS2014_5423}, it has been shown that for practical implementation the generator's loss should not include the $= \frac{1}{b}\left[\sum_{x \in \mathcal{B}_{r_i}}\log D_{\boldsymbol{\theta}_i}(x) \right]$ term since discriminator usually converges faster than the generator. Next, every IoTD sends this value to the center. Then, the center uses the received loss values from IoTDs to calculate its average loss:
\begin{align}
	L^g(\boldsymbol{\phi}) = \frac{1}{n}\sum_{i=1}^{n} L_i^g(\boldsymbol{\theta}_i).
\end{align}
Then, the center applies a gradient descent method on the generator to update its weight by minimizing	$ L^g(\boldsymbol{\phi}) $. 
\subsection{Discriminator Weight Swapping}
Every $ E $ epochs, every IoTD sends its discriminator's weights $ \boldsymbol{\theta}_i $ to the neighbor IoTD and receives the discriminator weights of another IoTD. Note that, as mentioned before, the connection graph should include a cycle covering all of the IoTDs. This helps IoTDs to receive the discriminator weights of all of the other IoTDs, eventually. This phase enables the the system to preserve privacy during training because the IoTDs do not share their data but, instead, they share the weights of their discriminators. Using this approach, the generator receives the loss value from all of the discriminators. Moreover, every IoTD's discriminator will be trained on all of the datasets. Therefore, for enough training epochs, the central generator will converge to the distribution of the total dataset $ \mathcal{D} $. In addition, the discriminators of the IoTDs will be similar to a GAN discriminator that has access to the total dataset. Thus, in the distributed GAN architecture, each IoTD's discriminator can detect intrusion on its own data as well as neighbor IoTDs. 
\subsection{Intrusion Detection Phase}
After the convergence of the distributed GAN, there will be no need for the central unit as all of the discriminators at IoTDs can detect the intrusion to the system. Thus, every IoTD IoTD will run its observed real-time data through its own discriminator as well as one of it's neighbor's discriminator. As shown in Proposition \ref{proposition}, the optimal discriminator will output $1/2$ for a normal state data point. Therefore, to detect an intrusion to the system, the output of the discriminator can be compared to $1/2$ and if the output is close to $1/2$, the IoTD will be in a normal state. However, if the output is closer to 0 or 1 the IoTD will be under attack. This approach helps the IoT system to detect an intrusion to the system without dependence on a central unit since every IoTD can also check its neighbor's data. The intrusion detection phase of the proposed distributed GAN-based IDS is shown in Fig. \ref{fig:systemmodel}.
\section{Simulation Results}\label{sec:simulations}
For our simulations, we use a daily activity recognition dataset \cite{dataset} that is collected from 30 subjects with different genders, ages, heights, and weights using a smartphone. We use this dataset since it can be a good example of health datasets that are collected by wearable IoTDs. Such health datasets are private to their owners and, thus, the owners may not intend to share them. The collected data is for 12 activities as follows: walking forward, walking left, walking right, walking upstairs, walking downstairs, running forward, jumping, sitting, standing, sleeping, elevator up, and elevator down. The dataset contains 2,365 recordings in total and every recording in this dataset has 561 frequency and time domain features. We split this dataset into training and test datasets with a 4 to 1 ratio.

We consider three IDSs: a standalone GAN that has access only to its own dataset, a centralized GAN that has access to all datasets, and the proposed distributed GAN. We use ANN architectures similar to the ones in \cite{NIPS2014_5423}. To train the ANNs we use the Tensorflow library, a single NVIDIA P100 GPU, and 20 Gigabits of memory. We train the ANNs on the training datasets and stop the training after 5,000 epochs. To model the intrusion, we consider an attacker who can inject false data to the features of the training dataset. We consider 10 scenarios with different attack-to-signal power ratios in which a random Gaussian noise is added to every feature of the training data points. In addition, we consider an internal attack detection in which every IoTD checks its own data and an external attack detection in which every IoTD checks its neighbor IoTDs.

\begin{figure*}[t!]
\centering
\begin{subfigure}{0.38\textwidth}
	\includegraphics[width=\columnwidth]{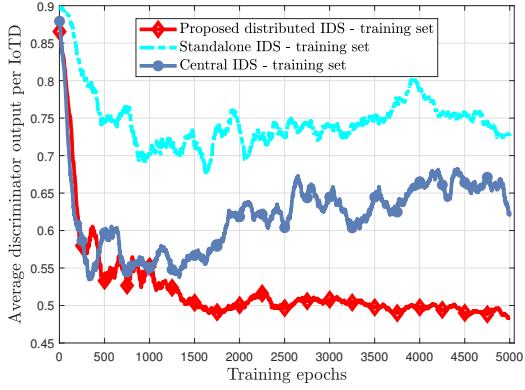}
	\caption{Average discriminator output of IDSs for training datasets.}
	\label{fig:conv_train}
\end{subfigure}
 ~
\begin{subfigure}{0.38\textwidth}
	\includegraphics[width=\columnwidth]{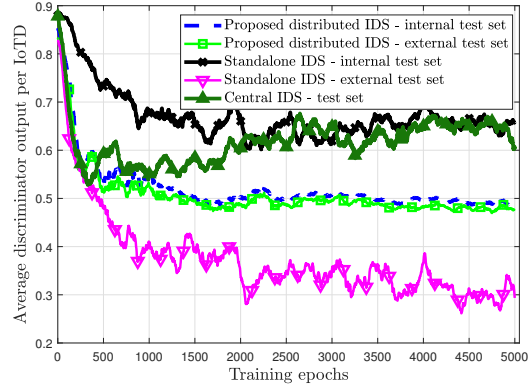}
	\caption{Average discriminator output of IDSs for test datasets.}
	\label{fig:conv_test}
\end{subfigure}
\caption{Average discriminator output VS training epochs}
\label{Fig:training}
\vspace{-6mm}
\end{figure*}

Fig. \ref{Fig:training} shows the output of the discriminator during the training phase. We can see from \ref{fig:conv_train} that the proposed IDS reaches the optimal discriminator value which is $0.5$ after convergence. However, the standalone and central IDSs reach values different from the optimal points. A standalone IDS has a non optimal discriminator output because it does not have access to the complete data. In addition, the central IDS also cannot reach the optimal discriminator since having access to all data points results in overfitting the ANN. As we can see from Fig. \ref{fig:conv_train} the central IDS's discriminator starts increasing around epoch 1000 after getting close to the optimal discriminator value. Such behavior of the central approach is aligned with the GAN-centric results of in \cite{hardy2018md}. Also Fig. \ref{fig:conv_test} shows that the proposed distributed IDS reaches the optimal discriminator even for the test dataset, although such dataset has not been used in the training. This shows that, the discriminators of the distributed GANs are successful in approximating the total dataset and can discriminate even the unseen data points.

\begin{figure}[t!]
	\centering
	\includegraphics[width=0.8\columnwidth]{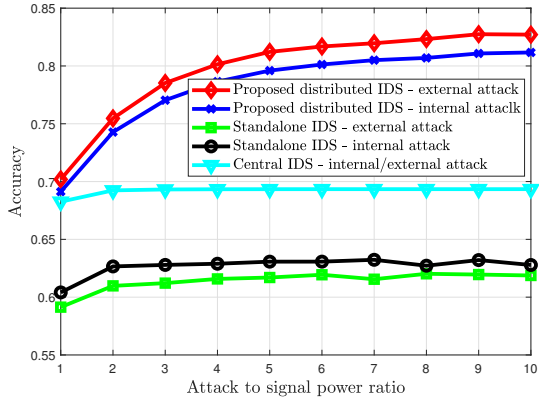}
	\caption{Accuracy of IDSs VS attack to signal power ratio.}
	\label{fig:accuracy}
	\vspace{-5mm}
\end{figure}

Fig. \ref{fig:accuracy} compares the accuracy of three IDSs against internal and external attacks. The accuracy is calculated as $\frac{T_P+T_N}{T_P+T_N + F_P+F_N}$, where $T_P$ is the true positive or the correctly detected anomaly ratio, $F_P$ is the false positive or the falsely detected anomaly ratio, $T_N$ is the true negative or the correctly assigned normal ratio, and $F_N$ is the false negative or the falsely assigned normal ratio. Fig. \ref{fig:accuracy} shows that the proposed distributed IDS outperforms both of the central and standalone IDSs by up to 15\% and 20\% in terms of accuracy, respectively, for both the internal and external attacks. Note that the internal and external attacks for a central IDS are the same since there is just one central unit that discriminates the IoTD data points.

\begin{figure}[t!]
	\centering
	\includegraphics[width=0.8\columnwidth]{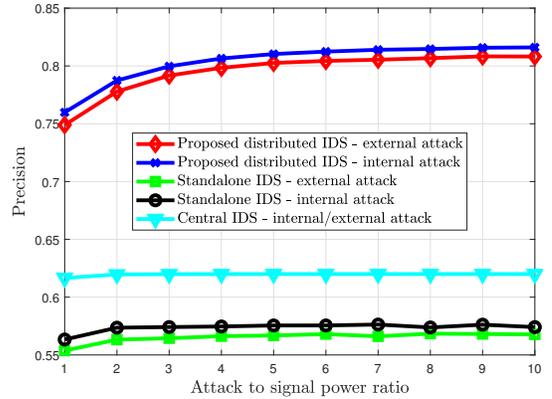}
	\caption{Precision of IDSs VS attack to signal power ratio.}
	\label{fig:precision}
	\vspace{-5mm}
\end{figure}

Fig. \ref{fig:precision} shows the precision of the three IDSs against internal and external attacks with different attack to signal ratios. In IDSs, precision is calculated as $ \frac{T_P}{T_P+F_P} $ and a precision score of 1.0 means that every datapoint labeled as intrusion is indeed an intrusion. We can see from Fig. \ref{fig:precision} that the proposed distributed IDS has up to 25\% and 20\% higher precision compared to the standalone IDS and the central IDS, respectively. This means that the proposed distributed IDS has higher confidence of detecting intrusion to the IoT.

\begin{figure}[t!]
	\centering
	\includegraphics[width=0.8\columnwidth]{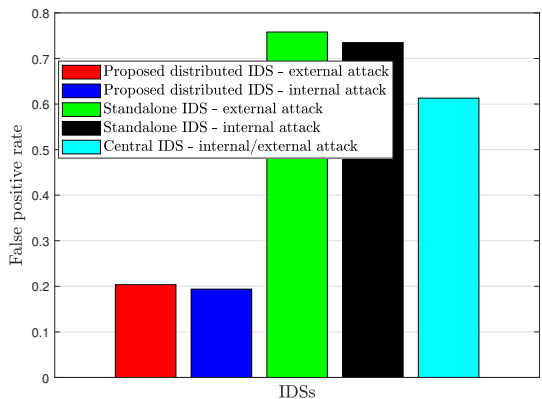}
	\caption{False positive ratio of IDSs.}
	\label{fig:FPR}
	\vspace{-6mm}
\end{figure}

Fig. \ref{fig:FPR} shows the false positive rate of the IDSs which is calculated as $ \frac{F_P}{F_P + T_N} $. In IDSs, a false positive rate is the ratio between the number of normal state data points wrongly categorized as intursion ($F_P$) and the total number of actual normal state data points ($ F_P + T_N $). Note that, since $ F_P $ and $ T_N $ are calculated from the evaluation of normal data points, the power ratio of the attack will not affect the false positive rate. We can see from Fig. \ref{fig:FPR} that the proposed distributed IDS has significantly less false positive rate than the standalone and central IDSs. In addition, both distributed and standalone IDSs have slightly lower false positive rate for internal attacks compared to external attacks.

\section{Conclusion}\label{sec:conc} 
In this paper, we have proposed a distributed GAN-based IDS solution that can detect intrusion to the IoT with minimum dependence on a central unit. In this architecture, every IoTD can monitor its own data as well as neighbor IoTDs to detect internal and external attacks. In addition, the proposed distributed IDS does not require sharing the datasets between the IoTDs, thus, it is practical to implement in IoTs that preserve the privacy of user data such as health monitoring systems or financial applications. We have shown analytically that the proposed distributed IDS can outperform a standalone IDS that has access only to the dataset of a single IDS. Simulation results that use a real-world a daily activity recognition dataset show that the proposed distributed GAN-based IDS has up tp 20\% higher accuracy, 25\% higher precision, and 60\% lower false positive rate compared to a standalone IDS.
\def\baselinestretch{0.97}
\bibliographystyle{IEEEtran}
\bibliography{references}
\end{document}